\documentclass[11pt]{article}
\usepackage{setspace}
\usepackage{amsfonts}
\usepackage{amsmath}
\usepackage{amssymb}
\usepackage{amsthm}
\usepackage{enumerate}
\usepackage{mathrsfs}

\numberwithin{equation}{section}
\newtheorem{theorem}{Theorem}[section]
\newtheorem{lemma}[theorem]{Lemma}

\newtheorem{corollary}[theorem]{Corollary}

\newenvironment{definition}[1][Definition.]{\begin{trivlist}
\item[\hskip \labelsep {\bfseries #1}]}{\end{trivlist}}

\newcommand{\threepartdef}[6]
{
	\left\{
		\begin{array}{lll}
			#1 & \mbox{if } #2 \\
			#3 & \mbox{if } #4 \\
			#5 & \mbox{if } #6
		\end{array}
	\right.
}

\begin{document}
\bibliographystyle{plain}
\title{Bounded Turing Reductions and Data Processing Inequalities for Sequences \thanks{This research was supported in part by National Science Foundation Grants 1247051 and 1545028. A preliminary version of part of this work was presented at the 11th International Conference on Computability, Complexity, and Randomness.}}
\author{Adam Case\\
Department of Computer Science\\
Iowa State University\\
Ames, IA 50011 USA}
\date{}
\maketitle

\begin{abstract}
A \emph{data processing inequality} states that the quantity of shared information between two entities (e.g. signals, strings) cannot be significantly increased when one of the entities is processed by certain kinds of transformations. In this paper, we prove several data processing inequalities for sequences, where the transformations are bounded Turing functionals and the shared information is measured by the lower and upper mutual dimensions between sequences.

We show that, for all sequences $X,Y,$ and $Z$, if $Z$ is computable Lipschitz reducible to $X$, then
\[
mdim(Z:Y) \leq mdim(X:Y) \text{ and } Mdim(Z:Y) \leq Mdim(X:Y).
\]
We also show how to derive different data processing inequalities by making adjustments to the computable bounds of the \emph{use} of a Turing functional.

The \emph{yield} of a Turing functional $\Phi^S$ with access to at most $n$ bits of the oracle $S$ is the smallest input $m \in \mathbb{N}$ such that $\Phi^{S \upharpoonright n}(m)\uparrow$. We show how to derive \emph{reverse data processing inequalities} (i.e., data processing inequalities where the transformation may significantly \emph{increase} the shared information between two entities) for sequences by applying computable bounds to the yield of a Turing functional.
\end{abstract}
\section{Introduction}
Various branches of information theory have developed methods for measuring the shared information between two objects. It is expected that a measure of mutual information satisfy certain properties (e.g., see \cite{jBell62}). Perhaps the most important property is the \emph{data processing inequality}, which says that the quantity of shared information between two objects cannot be significantly increased when one of the objects is processed by certain transformations \cite{bCovTho06}.

In algorithmic information theory, if $f:\Sigma^* \rightarrow \Sigma^*$ is a partial computable function, then there is a constant $c \in \mathbb{N}$ such that, for all strings $x,y \in \Sigma^*$,
\begin{align}\label{stringdpi}
I(f(x):y) \leq I(x:y) + c,
\end{align}
where $I(x:y) = K(y) - K(y\,|\,x)$ is the \emph{algorithmic mutual information} between strings $x$ and $y$ \cite{bLiVit08}. While (\ref{stringdpi}) is a data processing inequality for strings, there still exist settings within algorithmic information theory that do not have known data processing inequalities.

In this paper, we discuss several new data processing inequalities for sequences. We use \emph{mutual dimension}, a recent development in \emph{constructive dimension}, as the means for measuring the quantity of shared information between two sequences \cite{CasLut15,pCasLut15}. Lutz defined and explored the constructive dimension of sequences in \cite{jLutz03a}, and Mayordomo showed that constructive dimension can be characterized in terms of Kolmogorov complexity in \cite{jMayo02}. Mutual dimension is a generalization of constructive dimension and is defined in terms of algorithmic mutual information. Formally, the \emph{lower} and \emph{upper} mutual dimensions between sequences $S \in \Sigma^{\infty}$ and $T \in \Sigma^{\infty}$ are defined by
\[
mdim(S:T) = \displaystyle\liminf\limits_{n \rightarrow \infty}\frac{I(S \upharpoonright n:T \upharpoonright n)}{n\log|\Sigma|}
\]
and
\[
Mdim(S:T) = \displaystyle\limsup\limits_{n \rightarrow \infty}\frac{I(S \upharpoonright n:T \upharpoonright n)}{n\log|\Sigma|},
\]
respectively. Intuitively, these are the lower and upper \emph{densities} of algorithmic mutual information between $S$ and $T$. Originally, Case and Lutz defined the lower and upper mutual dimensions between points in Euclidean space and showed that they have \emph{all} of the expected properties that a measure of mutual information should have, including a data processing inequality \cite{CasLut15}. In a recent follow-up paper, the same authors extend this notion of mutual dimension to sequences and proved that it has several desirable properties \cite{pCasLut15}. However, no discussion regarding data processing inequalities for sequences was provided. Our primary goal in the present paper is to analyze how the lower and upper mutual dimensions between two sequences change when one of the sequences is transformed by a Turing functional.

A reduction can be described in several ways. Generally speaking, a problem $A$ reduces to a problem $B$ if $A$ is solvable when assuming that $B$ is solvable. In computability theory, Turing reductions are used to discuss the idea of relative computability. Formally, a sequence $S$ is \emph{Turing reducible} to a sequence $T$ if there exists an oracle machine that computes $S$ when $T$ is written on the oracle tape. We often refer to oracle machines as Turing functionals, which have been studied in detail by Rogers \cite{bRogers87} and Soare \cite{bSoare87,Soare09}. When a Turing functional $\Phi^S$ runs on a particular input, it is allowed to query the oracle $S$ at any time. The \emph{use} of a Turing functional is the largest position of the oracle tape that is queried during the computation of $\Phi^S$ on input $n$. We will be primarily concerned with Turing functionals whose use is bounded by a computable function.

Downey, Hirshfeldt, and LaForte first defined \emph{sw-reducibility} (strong weak truth table reducibility) as a Turing reduction whose use is bounded by $n + c$, where $n \in \mathbb{N}$ is the input and $c$ is a constant \cite{DoHiLa04}. The authors showed that, for all sequences $S$ and $T$, if $T$ is sw-reducible to $S$, then, for all $n \in \mathbb{N}$,
\begin{align*}
K(T \upharpoonright n) \leq K(S \upharpoonright n) + O(1).
\end{align*}
A sw-reduction is now referred to as a \emph{computable Lipschitz reduction} (\emph{cl-reduction}) because all Turing functionals whose use is bounded by $n + c$ can be viewed as an effective Lipschitz continuous function \cite{LewBar06a,LewBar06b}.

In section 3, we discuss data processing inequalities for sequences, where transformations are represented by Turing functionals with bounded use. Our main result of this section says that, for all sequences $X,Y,Z \in \Sigma^{\infty}$, if $Z$ is cl-reducible to $X$, then
\[
mdim(Z:Y) \leq mdim(X:Y)
\]
and
\[
Mdim(Z:Y) \leq Mdim(X:Y).
\]

We also show that, for all $\alpha \geq 1$, if $Z$ is reducible to $X$ via a functional $\Phi$ whose use is bounded by $\lceil \alpha(n + c) \rceil$, for all inputs $n \in \mathbb{N}$, then
\[
mdim(Z:Y) \leq \alpha \cdot mdim(X:Y)
\]
and
\[
Mdim(Z:Y) \leq \alpha \cdot Mdim(X:Y).
\]
We then provide weaker versions of the above inequalities stated in terms of the Turing functionals themselves.

In section 4, we explore \emph{reverse data processing inequalities} for sequences, i.e., data processing inequalities where the transformation may significantly \emph{increase} the amount of shared information between two objects. Unlike the data processing inequalities described above, we cannot derive reverse data processing inequalites by restricting how much of the oracle a Turing functional accesses. Instead, we place restrictions on the lengths of the strings that a Turing functional outputs.

In \cite{Gacs86}, G\'{a}cs analyzed the lengths of the outputs of monotonic operators, which are also used to describe Turing reductions. Similarly, we are interested in examining the lengths of the strings output by a Turing functional equipped with a finite oracle. We define the \emph{yield} of a Turing functional $\Phi^S$ with access to at most $n \in \mathbb{N}$ bits of the oracle $S$, denoted $\phi^S_{yield}(n)$, to be the smallest input $m \in \mathbb{N}$ such that $\Phi^{S \upharpoonright n}(m)\uparrow$. 

We say that a sequence $T$ is \emph{uniquely yield bounded reducible} (\emph{uyb-reducible}) to a sequence $S$ if there exists a Turing functional $\Phi$ such that,
\begin{enumerate}
\item if the first $\phi^S_{yield}(n)$ symbols of $\Phi^S$ is a prefix of $\Phi^T$, then the first $n$ symbols of $S$ is a prefix of $T$, and
\item $\phi^S_{yield}(n)$ is bounded by a computable function.
\end{enumerate}
Our main result of this section says that, for all sequences $X,Y,Z \in \Sigma^{\infty}$, if $Z$ is uyb-reducible to $X$ via a functional $\Phi$ such that $\phi^X_{yield}(n) \leq n + c$, for some constant $c \in \mathbb{N}$, then
\[
mdim(X:Y) \leq mdim(Z:Y)
\]
and
\[
Mdim(X:Y) \leq Mdim(Z:Y).
\]
We also show that, for all $\alpha \geq 1$, if $Z$ is uyb-reducible to $X$ via a functional $\Phi$ such that $\phi^X_{yield}(n) \leq \lceil \alpha(n + c) \rceil$, for all inputs $n \in \mathbb{N}$, then
\[
mdim(X:Y) \leq \alpha \cdot mdim(Z:Y)
\]
and
\[
Mdim(X:Y) \leq \alpha \cdot Mdim(Z:Y).
\]
\section{Preliminaries}
We begin by discussing several formal definitions and concepts related to Turing reductions, Kolmogorov complexity, and constructive dimension. Let $\mathbb{N} = \{0,1,2,\cdots\}$, $\Sigma=\{0,1,\ldots k-1\}$ be the \emph{alphabet} consisting of $k$ symbols, and $\Sigma^*$ be the set of all strings over $\Sigma$. We write $\Sigma^{\infty}$ for the set of all infinite sequences over $\Sigma$, and, for every $S \in \Sigma^{\infty}$ and $n \in \mathbb{N}$, $S[n]$ is the $n$th symbol of $S$ and $S \upharpoonright n$ denotes the first $n$ symbols of $S$. For all strings $x,y \in \Sigma^*$ and sequences $S \in \Sigma^{\infty}$, we write $x \sqsubseteq S$ and $x \sqsubseteq y$ to mean that $x$ is a \emph{prefix} of $S$ and $x$ is a prefix of $y$, respectively.

Oracle machines are used as a means of carrying out \emph{relative} computations, i.e., computations performed by Turing machines with access to an additional source of information provided by the oracle. An oracle machine is a Turing machine equipped with an additional read-only tape called the \emph{oracle tape}. We write $M^S$ to denote an oracle machine with sequence $S$ written on its oracle tape. Given an input $n \in \mathbb{N}$, an oracle machine will either halt or run forever. If the oracle machine halts on a given input, then it must query the oracle tape a finite number of times.

It is often useful to provide an oracle tape with a string rather than a sequence. The behavior of a machine $M$ with a string oracle $x \in \Sigma^*$ is identical to that of a sequence oracle $S \in \Sigma^{\infty}$, except that if the machine attempts to query a position of the oracle tape that is larger than $|x| - 1$, the machine immediately enters a looping state and runs forever.

The following notations and definitions can be found in \cite{uSpies11,bRogers87,Soare09}. We may disassociate an oracle machine $M$ from any particular oracle and refer to it as a partial function $\Phi_M: \Sigma^{\infty} \times \mathbb{N} \rightarrow \Sigma^*$ defined by $\Phi_M(S,n) = M^S(n)$. Each $\Phi_M$ is called a \emph{Turing functional}. The partial function $\Phi_M^S:\mathbb{N} \rightarrow \Sigma^*$ is defined by $\Phi_M^S(n) = \Phi_M(S,n)$, and we write $\Phi^S_M(n)\downarrow$ if $M^S$ halts on input $n$ and $\Phi^S_M(n)\uparrow$ if $M^S$ does not halt on input $n$.

For any two sequences $S$ and $T$ and any oracle machine $M$, we write $\Phi^S_M = T$ if, for all $n \in \mathbb{N}$,
\[
\Phi_M^S(n) = T \upharpoonright n.
\]

We say that $T$ is \emph{Turing reducible to} $S$ if there exists an oracle machine $M$ such that $\Phi_M^S = T$.

For the rest of this paper, we omit the $M$ in $\Phi_M$ and $\Phi_M^S$ and denote an arbitrary Turing functional by $\Phi$ and an arbitrary Turing functional with oracle $S$ by $\Phi^S$.

We now provide a brief overview of the basics of \emph{Kolmogorov complexity}. Specifically, we are interested in \emph{prefix-free Kolmogorov complexity}. Therefore, all Turing machines used in the following definitions will be self-delimiting.

Let $M$ be an arbitrary Turing machine. The \emph{conditional Kolmogorov complexity} of $x \in \Sigma^*$ \emph{given} $y \in \Sigma^*$ with respect to $M$ is
\[
K_M(x\,|\,y) = \min\{|\pi|\,\big|\,\pi \in \{0,1\}^* \text{ and } M(\pi,y) = x\}.
\]

The \emph{Kolmogorov complexity} of $x \in \Sigma^*$ with respect to $M$ is $K_M(x) = K_M(x\,|\,\lambda)$, where $\lambda$ is the \emph{empty string}. We say that a Turing machine $M'$ is \emph{optimal} if, for every Turing machine $M$, there is a constant $c_M \in \mathbb{N}$ such that, for all $x \in \Sigma^*$,
\[
K_{M'}(x) \leq K_M(x) + c_M,
\]
where $c_M$ is called an \emph{optimality constant} of $M$. An important fact in algorithmic information theory is that every universal Turing machine is optimal \cite{bLiVit08}. Therefore, we fix a particular universal Turing machine $U$ that we reference for the entirety of this paper and define the \emph{Kolmogorov complexity} of $x \in \Sigma^*$ by $K(x) = K_U(x)$ and the \emph{conditional} \emph{Kolmogorov complexity} of $x$ \emph{given} $y$ by $K(x\,|\,y) = K_U(x\,|\,y)$.

We define the \emph{joint Kolmogorov complexity} of $x \in \Sigma^*$ and $y \in \Sigma^*$ by $K(x,y) = K(\langle x,y \rangle)$, where $\langle \cdot \rangle$ is a string pairing function. The \emph{mutual information} between strings $x$ and $y$ is
\[
I(x:y) = K(y) - K(y\,|\,x),
\]
which is the quantity of algorithmic information that $x$ and $y$ share. For a more thorough discussion on this topic, see \cite{bLiVit08}.
\section{Turing Functionals with Bounded Use and Data Processing Inequalities}
In this section, we develop data processing inequalities for sequences and show how these inequalities change when applying different computable bounds to the \emph{use} of a Turing functional. First, we prove several supporting lemmas.

\begin{lemma}\label{cond}
There exists a constant $c \in \mathbb{N}$ such that, for all $u,v,w \in \Sigma^*$,
\[
K(u\,|\,vw) \leq K(u\,|\,v) + K(|v|) + c.
\]
\end{lemma}

\begin{proof}
Let $M$ be a TM such that, if $U(\pi_1) = |v|$ and $U(\pi_2,v) = u$,
\[
M(\pi_1\pi_2,vw) = u.
\]
Let $c_M \in \mathbb{N}$ be an optimality constant of $M$. Assume the hypothesis, and let $\pi_1$ be a minimum-length program for $|v|$ and $\pi_2$ be a minimum-length program for $u$ given $v$. By optimality,
\begin{align*}
K(u\,|\,vw) &\leq K_M(u\,|\,vw) + c_M\\
						&\leq |\pi_1\pi_2| + c_M\\
						&= K(u\,|\,v) + K(|v|) + c,
\end{align*}
where $c = c_M$. \qedhere
\end{proof}

\begin{corollary}\label{u:uv}
For all $u,v,w \in \Sigma^*$,
\[
I(u:w) \leq I(uv:w) + o(|u|).
\]
\end{corollary}

\begin{proof}
By the definition of mutual information and Lemma \ref{cond}, there exists a constant $c \in \mathbb{N}$ such that
\begin{align*}
I(u:w) &= K(w) - K(w\,|\,u)\\
			 &\leq K(w) - K(w\,|\,uv) + K(|u|) + c\\
			 &= I(uv:w) + o(|u|). \qedhere
\end{align*}
\end{proof}

The following lemma was proven in \cite{pCasLut15}.

\begin{lemma}\label{sym}
For all strings $u,w \in \Sigma^*$,
\[
I(u:w) = K(u) + K(w) - K(u,w) + o(|u|).
\]
\end{lemma}

\begin{corollary}\label{sym1}
For all $u,w \in \Sigma^*$,
\[
I(u:w) = I(w:u) + o(|u|) + o(|w|).
\]
\end{corollary}

\begin{proof}
By Lemma \ref{sym},
\begin{align*}
I(u:w) &= K(u) + K(w) - K(u,w) + o(|u|)\\
			 &= K(w) + K(u) - K(w,u) + o(|u|)\\
			 &= I(w:u) + o(|u|) + o(|w|). \qedhere
\end{align*}
\end{proof}

The following lemma was proven in \cite{CasLut15}.

\begin{lemma}\label{proc}
Let $f:\Sigma^* \times \Sigma^* \rightarrow \Sigma^*$ be a computable function. There exists a constant $c \in \mathbb{N}$ such that, for all strings $u,v,w \in \Sigma^*$,
\[
K(u\,|\,w) \leq K(u\,|\,f(w,v)) + K(v) + c.
\]
\end{lemma}

We now investigate bounded Turing reductions and their effects on the shared algorithmic information between strings. As previously mentioned, a halting oracle machine computation can only make a finite number of queries to its oracle, and we are often interested in knowing the largest position of the oracle tape that a machine will query before it halts. The following definition is from \cite{uSpies11}.

\begin{definition} The \emph{use function} of a Turing functional $\Phi$ equipped with oracle $S \in \Sigma^{\infty}$ is
\[
\phi^S_{use}(n) = \threepartdef{m + 1}{\Phi^S(n)\downarrow \text{ and } m \text{ is the largest query made to the oracle $S$}}{0}{\Phi^S(n)\downarrow \text{ and the oracle $S$ is not queried during the computation}}{\text{undefined}}{\Phi^S(n)\uparrow},
\]
for every $n \in \mathbb{N}$.
\end{definition}
We denote Turing functionals using uppercase Greek letters (e.g., $\Phi$, $\Gamma$) and their corresponding use functions by lowercase Greek letters (e.g., $\phi_{use}$, $\gamma_{use}$).

\begin{definition}
A sequence $T \in \Sigma^{\infty}$ is \emph{bounded Turing reducible} (\emph{bT-reducible}) to a sequence $S \in \Sigma^{\infty}$ if $T$ is Turing reducible to $S$ by a Turing functional $\Phi$ such that $\phi^S_{use}$ is bounded by a computable function.
\end{definition}

For convenience, we say that $T \in \Sigma^{\infty}$ is \emph{m-bT-reducible} to $S \in \Sigma^{\infty}$ if $T$ is bT-reducible to $S$ via $\Phi$ and $m: \mathbb{N} \rightarrow \mathbb{N}$ is a computable function bounding $\phi^S_{use}$.

\begin{lemma}\label{tech}
Let $m: \mathbb{N} \rightarrow \mathbb{N}$ be an increasing, computable function. For all $X,Y,Z \in \Sigma^{\infty}$, if $Z$ is $m$-bT-Turing reducible to $X$, then
\[
I(Z \upharpoonright n:Y\upharpoonright n) \leq I(X\upharpoonright m(n):Y\upharpoonright m(n)) + o(m(n)).
\]
\end{lemma}

\begin{proof}
Assume that $Z$ is $m$-bT-Turing reducible to $X$ by some Turing functional $\Phi$ whose use function $\phi^X_{use}$ is bounded by $m$. By Corollaries \ref{u:uv} and \ref{sym1},
\begin{align}\label{part1}
I(Z \upharpoonright n:Y \upharpoonright n) &= I(Y \upharpoonright n: Z \upharpoonright n) + o(n) \nonumber \\
																					 &\leq I(Y \upharpoonright m(n): Z \upharpoonright n) + o(n)\\
																					 &= I(Z \upharpoonright n: Y \upharpoonright m(n)) + o(m(n)). \nonumber
\end{align}
Define the partial function $f: \{0,1\}^* \times \mathbb{N} \rightarrow \{0,1\}^*$ by
\[
f(u,n) = \Phi^{u}(n),
\]
for all $u \in \Sigma^*$ and $n \in \mathbb{N}$. The function $f$ is clearly computable. Therefore, by (\ref{part1}) and Lemma \ref{proc},
\begin{align*}
I(Z \upharpoonright n: Y \upharpoonright n) &\leq I(f(X \upharpoonright m(n), n):Y \upharpoonright m(n)) + o(m(n))\\
																						&= K(Y \upharpoonright m(n)) - K(Y \upharpoonright m(n)\,|\,f(X \upharpoonright m(n), n)) + o(m(n))\\
																						&\leq K(Y \upharpoonright m(n)) - K(Y \upharpoonright m(n)\,|\,X \upharpoonright m(n)) + o(m(n))\\
																						&= I(X \upharpoonright m(n):Y \upharpoonright m(n)) + o(m(n)). \qedhere
\end{align*}
\end{proof}

The first notion of mutual dimension was defined in \cite{CasLut15} to analyze the density of algorithmic mutual information between points in Euclidean space. It was then extended to sequences in \cite{pCasLut15} in order to study coupled randomness.

\begin{definition}
The \emph{lower} and \emph{upper mutual dimensions} between $S \in \Sigma^{\infty}$ and $T \in \Sigma^{\infty}$ are
\[
mdim(S:T) = \displaystyle\liminf\limits_{n \rightarrow \infty}\frac{I(S \upharpoonright n:T\upharpoonright n)}{n\log\,|\Sigma|}
\]
and
\[
Mdim(S:T) = \displaystyle\limsup\limits_{n \rightarrow \infty}\frac{I(S\upharpoonright n:T\upharpoonright n)}{n\log\,|\Sigma|},
\]
\end{definition}
respectively. We now present an important technical lemma.

\begin{lemma}[Bounded Use Processing Lemma]\label{proc lemma}
Let $m: \mathbb{N} \rightarrow \mathbb{N}$ be an increasing, computable function. For all $X,Y,Z \in \Sigma^{\infty}$, if $Z$ is $m$-bT-Turing reducible to $X$, then
\[
mdim(Z:Y) \leq mdim(X:Y)\bigg (\displaystyle\limsup\limits_{n \rightarrow \infty}\frac{m(n)}{n} \bigg)
\]
and
\[
Mdim(Z:Y) \leq Mdim(X:Y)\bigg (\displaystyle\limsup\limits_{n \rightarrow \infty}\frac{m(n)}{n} \bigg),
\]
except when $\bigg (\displaystyle\limsup\limits_{n \rightarrow \infty}\frac{m(n)}{n} \bigg) = \infty$ while either $mdim(X:Y) = 0$ or $Mdim(X:Y) = 0$.
\end{lemma}

\begin{proof}
By Lemma \ref{tech},
\begin{align*}
mdim(Z:Y) &= \displaystyle\liminf\limits_{n \rightarrow \infty}\frac{I(Z \upharpoonright n: Y \upharpoonright n)}{n\log\,|\Sigma|}\\
					&\leq \displaystyle\liminf\limits_{n \rightarrow \infty}\frac{I(X \upharpoonright m(n): Y \upharpoonright m(n)) + o(m(n))}{n\log\,|\Sigma|}\\
					&= \displaystyle\liminf\limits_{n \rightarrow \infty}\bigg (\frac{I(X \upharpoonright m(n): Y \upharpoonright m(n)) + o(m(n))}{m(n)\log|\Sigma|} \cdot \frac{m(n)}{n} \bigg )\\
					&\leq \bigg ( \displaystyle\liminf\limits_{n \rightarrow \infty}\frac{I(X \upharpoonright m(n): Y \upharpoonright m(n)) + o(m(n))}{m(n)\log|\Sigma|}\bigg ) \bigg( \displaystyle\limsup\limits_{n \rightarrow \infty} \frac{m(n)}{n} \bigg )\\
					&= mdim(X:Y)\bigg (\displaystyle\limsup\limits_{n \rightarrow \infty}\frac{m(n)}{n} \bigg ).
\end{align*}
A similar proof can be given for $Mdim$. \qedhere
\end{proof}

\begin{definition}
Let $m: \mathbb{N} \rightarrow \mathbb{N}$ be defined by $m(n) = n + c$, where $c \in \mathbb{N}$ is a constant. A sequence $T \in \Sigma^{\infty}$ is \emph{computable Lipschitz reducible} ($cl$-\emph{reducible}) to a sequence $S \in \Sigma^{\infty}$ if $T$ is $m$-bT-reducible to $S$.
\end{definition}


The following theorem follows directly from the Bounded Use Processing Lemma.

\begin{theorem}\label{dpi}
For all sequences $X,Y,Z \in \Sigma^{\infty}$, if $Z$ is cl-reducible to $X$, then
\[
mdim(Z:Y) \leq mdim(X:Y)
\]
and
\[
Mdim(Z:Y) \leq Mdim(X:Y).
\]
\end{theorem}

Let $\alpha \geq 1$ and $h_{\alpha}: \mathbb{N} \rightarrow \mathbb{N}$ be defined by $h_{\alpha}(n) = \lceil \alpha(n + c) \rceil$, where $c \in \mathbb{N}$ is a constant. The following is a corollary of the Bounded Use Processing Lemma.

\begin{corollary}\label{gbound}
Let $\alpha \geq 1$. For all sequences $X,Y,Z \in \Sigma^{\infty}$, if $Z$ is $h_{\alpha}$-bT-reducible to a sequence $X$, then
\[
mdim(Z:Y) \leq \alpha \cdot mdim(X:Y)
\]
and
\[
Mdim(Z:Y) \leq \alpha \cdot Mdim(X:Y).
\]
\end{corollary}

Typically, data processing inequalities are statements about \emph{all} of the defined outputs of a particular transformation. The results above, while strong, are not framed in this manner. To remedy this, we now discuss data processing inequalities in terms of individual bounded Turing functionals.

\begin{definition}
Let $m:\mathbb{N} \rightarrow \mathbb{N}$ be a computable function. A $m$-\emph{bounded Turing functional} (\emph{m-bT-functional}) is a Turing functional such that, for every sequence $S \in \Sigma^{\infty}$ and every $n \in \mathbb{N}$ where $\Phi^S(n)$ is defined, $\phi^S_{use}(n) \leq m(n)$.
\end{definition}


%

\begin{definition}
Let $m: \mathbb{N} \rightarrow \mathbb{N}$ be defined by $m(n) = n + c$. A $cl$-\emph{functional} is a $m$-bounded Turing functional.
\end{definition}
%

We use Theorem \ref{dpi} and Corollary \ref{gbound} to derive the following data processing inequalities for sequences whose transformations are bounded Turing functionals.

\begin{corollary}
If $\Phi$ is a cl-functional, then, for all $S,T \in \Sigma^{\infty}$ where $\Phi^S$ is defined,
\[
mdim(\Phi^S:T) \leq mdim(S:T)
\]
and
\[
Mdim(\Phi^S:T) \leq Mdim(S:T).
\]
\end{corollary}

We also have a similar data processing inequality for $h_{\alpha}$-bounded Turing functionals.
\begin{corollary}
For all $\alpha \geq 1$, if $\Phi$ is a $h_{\alpha}$-bounded Turing functional, then, for all $S,T \in \Sigma^{\infty}$ where $\Phi^S$ is defined,
\[
mdim(\Phi^S:T) \leq \alpha \cdot mdim(S:T)
\]
and
\[
Mdim(\Phi^S:T) \leq \alpha \cdot Mdim(S:T).
\]
\end{corollary}
\section{Turing Functionals with Bounded Yield and Reverse Data Processing Inequalities}
In this section, we define the \emph{yield} of a Turing functional and develop several reverse data processing inequalities (i.e., data processing inequalities where the transformations may significantly \emph{increase} the mutual dimension between two sequences) using yield bounded Turing functionals.

We now introduce the \emph{yield function} of a Turing functional.
\begin{definition}
The \emph{yield function} of a Turing functional $\Phi$ equipped with oracle $S \in \Sigma^{\infty}$ is defined by
\[
\phi^S_{yield}(n) = \min\{m \in \mathbb{N} \,|\, \Phi^{S\upharpoonright n}(m)\uparrow\},
\]
for all $n \in \mathbb{N}$.
\end{definition}

Intuitively, ``use'' is how much of the oracle the Turing functional must access in order for it to halt on a given input, and ``yield'' is how many inputs the Turing functional can halt on given a prefix of the oracle.

\begin{definition}
A sequence $T \in \Sigma^{\infty}$ is \emph{yield bounded reducible} (\emph{yb-reducible}) to a sequence $S \in \Sigma^{\infty}$ if $T$ is Turing reducible to $S$ by a Turing functional $\Phi$ such that $\phi^S_{yield}$ is bounded by a computable function.
\end{definition}

For convenience, we say that $T$ is $m$\emph{-yb-reducible} to $S$ if $T$ is yb-reducible to $S$ and $m: \mathbb{N} \rightarrow \mathbb{N}$ is a computable function bounding $\phi^S_{yield}$.

In order to develop reverse data processing inequalities for sequences, we need to apply the following restriction to our Turing functionals.

\begin{definition}
A Turing functional $\Phi^S$ is \emph{uniquely yielding} for an oracle $S \in \Sigma^{\infty}$ if, for all $T \in \Sigma^{\infty}$ and $n \in \mathbb{N}$,
\[
\Phi^S\upharpoonright \phi^S_{yield}(n) \sqsubseteq \Phi^T \Rightarrow S \upharpoonright n \sqsubseteq T.
\]
\end{definition}

\begin{definition}
A sequence $T \in \Sigma^{\infty}$ is \emph{uniquely yield bounded reducible} (\emph{uyb-reducible}) to $S \in \Sigma^{\infty}$ if $T$ is yb-reducible to $S$ by a Turing functional that is uniquely yielding.
\end{definition}

We say that $T$ is $m$\emph{-uyb-reducible} to $S$ if $T$ is uyb-reducible to $S$ by a Turing functional whose yield function is bounded by a computable function $m: \mathbb{N} \rightarrow \mathbb{N}$.

\begin{lemma}\label{rev tech}
If $T \in \Sigma^{\infty}$ is $m$-uyb-reducible to $S \in \Sigma^{\infty}$, then $S$ is $m$-bT-reducible to $T$.
\end{lemma}

\begin{proof}
Let $T$ be $m$-uyb-reducible to $S$ by a Turing functional $\Phi$. We define a Turing functional $\Gamma^T$ that operates on an input $n \in \mathbb{N}$ by querying the first $m(n)$ bits of $T$ and searching for a string $x \in \Sigma^*$ such that $|x| \geq n$ and $\Phi^x(m(n)) = T \upharpoonright m(n)$. After finding $x$, $\Gamma^T$ outputs $x \upharpoonright n$. Observe that
\begin{align*}
\Phi^S \upharpoonright \phi^S_{yield}(n) &\sqsubseteq \Phi^S \upharpoonright m(n)\\
													&= T \upharpoonright m(n)\\
													&= \Phi^x(m(n))\\
													&\sqsubseteq \Phi^x.
\end{align*}
Since $\Phi$ is uniquely yielding for $S$ and $|x| \geq n$, $S \upharpoonright n \sqsubseteq x$, which implies that $\Gamma^T(n) = S \upharpoonright n$.
\end{proof}

The following lemma follows directly by the Bounded Use Processing Lemma and Lemma \ref{rev tech}.

\begin{lemma}[Bounded Yield Processing Lemma]
Let $m: \mathbb{N} \rightarrow \mathbb{N}$ be a increasing, computable function. For all $X,Y,Z \in \Sigma^{\infty}$, if $Z$ is $m$-uyb-reducible to $X$, then
\[
mdim(X:Y) \leq mdim(Z:Y)\bigg (\displaystyle\limsup\limits_{n \rightarrow \infty}\frac{m(n)}{n} \bigg)
\]
and
\[
Mdim(X:Y) \leq Mdim(Z:Y)\bigg (\displaystyle\limsup\limits_{n \rightarrow \infty}\frac{m(n)}{n} \bigg),
\]
except when $\bigg (\displaystyle\limsup\limits_{n \rightarrow \infty}\frac{m(n)}{n} \bigg) = \infty$ while either $mdim(Z:Y) = 0$ or $Mdim(Z:Y) = 0$.
\end{lemma}

\begin{definition}
Let $m: \mathbb{N} \rightarrow \mathbb{N}$ be defined by $m(n) = n + c$. A sequence $T \in \Sigma^{\infty}$ is \emph{linear uniquely yield bounded reducible} ($\ell$\emph{-uyb-reducible}) to a sequence $S \in \Sigma^{\infty}$ if $T$ is $m$-uyb-reducible to $S$.
\end{definition}

The following theorem and corollary follow directly from the Bounded Yield Processing Lemma.

\begin{theorem}\label{rdpi}
For all sequences $X,Y,Z \in \Sigma^{\infty}$, if $Z$ is $\ell$-uyb-reducible to $X$, then
\[
mdim(X:Y) \leq mdim(Z:Y)
\]
and
\[
Mdim(X:Y) \leq Mdim(Z:Y).
\]
\end{theorem}

\begin{corollary}\label{rgbound}
Let $\alpha \geq 1$. For all sequences $X,Y,Z \in \Sigma^{\infty}$, if $Z$ is $h_{\alpha}$-uyb-reducible to $X$, then
\[
mdim(X:Y) \leq \alpha \cdot mdim(Z:Y)
\]
and
\[
Mdim(X:Y) \leq \alpha \cdot Mdim(Z:Y).
\]
\end{corollary}

The end of Section 3 discussed data processing inequalities in terms of the defined outputs of use bounded Turing functionals. In like manner, we describe reverse data processing inequalities in terms of yield bounded Turing functionals.

\begin{definition}
A Turing functional is a \emph{yield bounded functional} (\emph{yb-functional}) if there exists a computable function $f: \mathbb{N} \rightarrow \mathbb{N}$ such that, for every $S \in \Sigma^{\infty}$, $\phi^S_{yield}(n) \leq f(n)$.
\end{definition}


\begin{definition}
A \emph{uniquely yield bounded functional} (\emph{uyb-functional}) is a yield bounded functional that is also uniquely yielding for every oracle.
\end{definition}

For convenience, we say that a Turing functional is a $m$\emph{-uyb-functional} if it is a uyb-functional whose yield is bounded by a computable function $m:\mathbb{N} \rightarrow \mathbb{N}$.


\begin{definition}
Let $m: \mathbb{N} \rightarrow \mathbb{N}$ be defined by $m(n) = n + c$. A Turing functional is a \emph{linear uniquely yield bounded functional} ($\ell$\emph{-uyb-functional}) if it is a $m$-uyb-functional.
\end{definition}

We use Theorem \ref{rdpi} and Corollary \ref{rgbound} to derive the following reverse data processing inequalities for sequences whose transformations are uniquely yield bounded Turing functionals.

\begin{corollary}
For all $\ell$-uyb-functionals $\Phi$ and sequences $S,T \in \Sigma^{\infty}$ where $\Phi^S$ is defined,
\[
mdim(S:T) \leq mdim(\Phi^S:T)
\]
and
\[
Mdim(S:T) \leq Mdim(\Phi^S:T).
\]
\end{corollary}

\begin{corollary}
Let $\alpha \geq 1$. For all $h_{\alpha}$-uyb-functionals $\Phi$ and sequences $S,T \in \Sigma^{\infty}$ where $\Phi^S$ is defined,
\[
mdim(S:T) \leq \alpha \cdot mdim(\Phi^S:T)
\]
and
\[
Mdim(S:T) \leq \alpha \cdot Mdim(\Phi^S:T).
\]
\end{corollary}
\section*{Acknowledgments}
The author would like to thank Xiang Huang, Jack Lutz, Timothy McNicholl, and Donald Stull for useful discussions.

\bibliography{Master}

\end{document}